\documentclass[letterpaper, 10 pt, conference]{ieeeconf}

\IEEEoverridecommandlockouts                              

\usepackage{xcolor}
\usepackage{amsmath, amssymb, amsfonts, mathtools, mathrsfs}
\usepackage{dsfont}        
\usepackage{amsthm}        

\usepackage{algorithm}
\usepackage{algorithmic}
\newcommand{\INPUT}{\STATE \textbf{Input:}}
\newcommand{\OUTPUT}{\STATE \textbf{Output:}}

\usepackage{graphicx}
\usepackage[caption=false,font=footnotesize]{subfig}
\usepackage{epstopdf}      
\usepackage{makecell}      

\usepackage{textcomp}
\usepackage{xcolor}

\definecolor{algcomment}{rgb}{0.35,0.35,0.35}
\newcommand{\ALGCOMMENT}[1]{{\color{algcomment}\footnotesize\ (// #1)}}

\usepackage{setspace}
\usepackage{color,soul}    

\usepackage{cite}          
\usepackage{url}
\usepackage[
    colorlinks=true,
    linkcolor=blue,
    citecolor=green,
    urlcolor=magenta
]{hyperref}

\usepackage{comment}       

\DeclarePairedDelimiter{\norm}{\lVert}{\rVert}

\theoremstyle{definition}

\newtheorem{proposition}{Proposition}
\newtheorem{lemma}{Lemma}
\newtheorem{remark}{Remark}

\DeclareMathOperator{\blkdiag}{blkdiag}

\newcommand{\ie}{\textit{i.e.}}

\allowdisplaybreaks

\title{Reduced Order Observers for Monocular Visual Inertial Odometry} 
\author{Amel Abi, Mouaad Boughellaba, Miaomiao Wang and Abdelhamid Tayebi  
	\thanks{This work was supported by the Natural Sciences and Engineering Research Council of Canada (NSERC), under the grant RGPIN-2026-06037.}
	\thanks{
    A. Abi, M. Boughellaba and A. Tayebi are with the Department of Electrical and Computer Engineering, Lakehead University, Thunder Bay, ON P7B 5E1, Canada. (e-mail: {\tt\small aabi,mboughel,atayebi@lakeheadu.ca}). M. Wang is with the School of Artificial Intelligence and Automation, Huazhong University of Science and Technology, Wuhan 430074, China (e-mail: {\tt\small mmwang@hust.edu.cn})}%
}
\begin{document}
\maketitle
\raggedbottom



\begin{abstract}
This work develops a reduced order observer framework of monocular visual-inertial odometry (VIO). The key idea consists in considering the gravity vector in the body-frame as an additional state and constructing an appropriate output that relates the monocular bearing measurements and their time derivatives to the body-frame velocity. This eliminates the need for the estimation of the landmark positions, leading to a six-dimensional linear time-varying (LTV) observer for the body-frame velocity and gravity, endowed with uniform global exponential stability guarantees. These estimates are subsequently used to recover the orientation almost globally, up to an unknown constant yaw offset, and the position up to an unknown constant translation. The mixed-bearing framework is further extended to account for the inertial measurement unit (IMU) biases, yielding a twelve-dimensional LTV observer endowed with local exponential stability guarantees. Numerical simulations are provided to illustrate the effectiveness of the proposed observers.
\end{abstract}


\section{Introduction}

 The growing use of autonomous unmanned aerial vehicles (UAVs) in GPS-denied environments has intensified the need for robust onboard localization. Visual-inertial odometry, which fuses camera and IMU data, provides a viable solution to this requirement. Although
 VIO is widely adopted in applications such as augmented reality and virtual reality, as well as on smartphones, its deployment for aerial navigation is constrained by the limited onboard computational resources of UAVs. Current VIO methodologies typically involve a trade-off between computational efficiency and estimation accuracy. Extended Kalman Filter (EKF)-based approaches, such as MSCKF \cite{MSCKF2007}, OpenVINS \cite{OpenVINS}, and ROVIO \cite{ROVIO}, offer the computational efficiency required for real-time onboard operation. However, their reliance on local linearization makes them susceptible to linearization errors. In contrast, optimization-based methods such as OKVIS \cite{OKVIS} and VINS-Mono \cite{VINS-mono} mitigate these errors through iterative nonlinear optimization over a sliding window. This generally yields higher estimation accuracy, at the expense of increased computational cost.
 
 A key limitation of EKF-based approaches is filter inconsistency. The VIO problem possesses four unobservable degrees of freedom: the three components of global position and the rotation about the gravity direction (yaw).  Linearization can introduce spurious information along these unobservable directions, rendering the filter inconsistent \cite{Scaramuzza2020}. 
To address this source of inconsistency, invariant and equivariant observer designs have been proposed. These approaches exploit Lie group symmetries of the VIO  problem. The Invariant Extended Kalman Filter (IEKF) \cite{IEKF} utilizes the extended special Euclidean group $\mathrm{SE}_2(3)$ for the inertial states. Under this symmetry, the bias-free IMU dynamics are group-affine, yielding a linearization independent of the estimated trajectory. EqVIO \cite{Eqvio} extends this principle within the Equivariant Filter (EqF) framework by incorporating the Scaled Orthogonal Transformations group $\mathrm{SOT}(3)$ for visual landmarks, enabling a higher-order approximation of the visual measurement function while reducing linearization errors and improving filter consistency. While these methods address the inconsistency problem, their convergence guarantees remain local. Furthermore, these estimation frameworks require augmenting the state vector with visual landmark estimates, leading to computational complexity that grows with the number of tracked landmarks.
This motivates the need for a reduced-order estimator whose state dimension is independent of the number of tracked landmarks, while achieving strong (non-local) asymptotic stability guarantees. In \cite{Bouazza}, a reduced-order cascaded observer was developed for monocular VIO. The proposed observer does not incorporate bearing measurements directly in the output equation, but rather uses the direction of the body-frame linear velocity, which is extracted from optical flow data via a constrained minimization on the unit sphere. 

In this work, we propose a framework for pose and linear velocity estimation of a rigid body using monocular vision and inertial measurements. By augmenting the translational state with the gravity vector expressed in the body frame, the translational dynamics can be recast as a linear time-varying (LTV) system that does not explicitly depend on the attitude dynamics. In contrast to \cite{Bouazza}, our reduced-order observers employ a mixed-bearing strategy directly within the output equation, eliminating the need for an optimization step to estimate the direction of the body-frame linear velocity.
 

We propose three estimation schemes. The first employs an LTV observer evolving on $\mathbb{R}^{6+3N}$ to jointly estimate the body-frame linear velocity, gravity vector, and positions of $N$ landmarks, with uniform global exponential stability (UGES) guarantees. The estimated velocity and gravity vector are then used to estimate the rigid body's position and orientation.
In this formulation, the observer dimension grows with the number of tracked landmarks. To address this limitation, the second scheme employs a reduced-order observer evolving on $\mathbb{R}^6$, which uses mixed-bearing information to estimate the body-frame linear velocity and gravity vector without explicitly estimating landmark positions, while retaining UGES guarantees. These estimates are subsequently used to estimate the rigid body's position and orientation. Both schemes recover the orientation almost globally, up to an unknown constant yaw offset, and the position up to an unknown constant translation.
Finally, we extend the mixed-bearing formulation to account for gyroscope and accelerometer biases, resulting in an LTV observer on $\mathbb{R}^{12}$ with local exponential stability guarantees. As in the first two schemes, the estimated biases, body-frame linear velocity, and gravity vector are used to recover the rigid body's orientation up to an unknown constant yaw offset and its position up to an unknown constant translation.

\section{Notations and preliminaries}

We denote by $\mathbb{R}$ and $\mathbb{N}$ the set of real numbers and natural numbers, respectively. Let $I_n$ and $0_{m\times n}$ refer to the $n\times n$ identity matrix and the $m\times n$ zero matrix, respectively. The column vector of ones is denoted by $\mathbf{1}_N \in \mathbb{R}^N$. Let $\norm{\cdot}$ represent the Euclidean norm. For $x\in\mathbb{R}^3$, $x^{\times}\in\mathbb{R}^{3\times 3}$ denotes the skew-symmetric matrix such that $x^\times y = x\times y$ for all $y\in\mathbb{R}^3$, where $\times$ indicates the cross product. For matrices $A$ and $B$, $A \otimes B$ represents their Kronecker product.
The special orthogonal group is defined as $\mathrm{SO}(3):=\{R\in\mathbb{R}^{3\times 3}\mid R^\top R=I_3,\det(R)=1\}$. The set of three-dimensional unit vectors is given by $\mathbb{S}^2:=\{
{x}\in\mathbb{R}^3\mid\|{x}\| = 1\}$. For any $x\in \mathbb{S}^2$, define the orthogonal projection operator $\pi:\mathbb{S}^2\rightarrow\mathbb{R}^{3\times3}$ by
\begin{equation}
    \pi(x) = I_3 - xx^\top.
\end{equation}
The operator $\pi(x)$ projects any vector in $\mathbb{R}^3$ onto the plane orthogonal to $x$. 

\subsection{Riccati equation and uniform observability}
Consider the following Riccati equation:
\begin{equation}\label{eq:Riccati_equation}
    \dot{P}= A(t)P+PA(t)^ \top - PC^\top( t) Q^{-1}(t)C(t)P+ V(t),
\end{equation}
with $P(0)=P^\top(0)>0$, $A(t)$, $C(t)$ are continuous and $Q(t)$, $V(t)$ are continuous and symmetric, satisfying
\[
q_m I\le Q(t)\le q_M I,
\qquad
v_m I\le V(t)\le v_M I,
\]
for some strictly positive constants $v_m$, $v_M$, $q_m$ and $q_M$. Let $\Phi(t,t_0)$ denotes the state-transition matrix of the system $\dot{x}=A(t)x$.  Define the observability Gramian
 \begin{equation}\label{obs_gram}
 W_0(t,t+\delta)= \int_{t}^{t+\delta} \Phi^{\top}(\tau,t) C^{\top}(\tau) C(\tau) \Phi(\tau,t) d\tau,
 \end{equation}
 and 
 \begin{equation}
 W_c(t,t+\sigma)= \int_{t}^{t+\sigma} \Phi(t,\tau) \Phi^\top (t,\tau) d\tau.
\end{equation}
A simplified version of the results in \cite{Bucy1972}, in view of the properties imposed on $V(t)$ and $Q(t)$, is stated in the following lemma:
\begin{lemma}\label{lemma1}
If there exist positive constants $\delta$, $\sigma$, $\mu_o$, $\bar{\mu}_o$, $\mu_c$, $\bar{\mu}_c$ such that 
\begin{equation}\label{UCO}
\bar{\mu}_o I \geq  W_0(t,t+\delta)\ge \mu_o I,~ \forall t\geq 0,
\end{equation}
and 
\begin{equation}\label{UCOplus}
\bar{\mu}_c I \geq  W_c(t,t+\sigma)\ge \mu_c I, ~ \forall t\geq 0.
\end{equation}
Then the Riccati equation \eqref{eq:Riccati_equation} admits a solution $P(t)$ satisfying
\begin{equation}\label{Riccati_bounds}
P_m I\leq P(t)\leq P_M I,~~\forall t\geq 0,
\end{equation}
for some strictly positive scalars $P_m$ and $P_M$.
\end{lemma}

A simplified version of Lemma \ref{lemma1} in view of the uniform boundedness of $A(t)$ and $C(t)$ is given in the following lemma:
\begin{lemma}\label{lemma2}
If $A(t)$ and $C(t)$ are continuous and uniformly bounded and if there exist positive constants $\delta$ and $\mu$ such that 
\begin{equation}\label{UO}
W_0(t,t+\delta)\ge \mu I,~\forall t\geq 0.
\end{equation}
Then the Riccati equation \eqref{eq:Riccati_equation} admits a solution $P(t)$ satisfying \eqref{Riccati_bounds}.
\end{lemma}
Condition \eqref{UCO}, which imposes uniform lower and upper bounds on the observability Gramian, is commonly referred to as the uniform complete observability condition \cite{Kalman1960,Bucy1972}. In contrast, the one-sided condition \eqref{UO}, requiring only a uniform positive lower bound on the observability Gramian, is often referred to as uniform observability in more recent literature; see, for instance, \cite{Morin2017,HAMEL2017}. 
The terminology, however, is not entirely consistent across the literature. In particular, some authors use the term uniform complete observability for the one-sided condition \eqref{UO}, see for instance, \cite{BATISTA2017}. This discrepancy is largely terminological rather than conceptual, since in many commonly considered settings the upper bound on the observability Gramian follows from additional regularity or boundedness assumptions on the system matrices. Nevertheless, when such an upper bound is not guaranteed a priori, the distinction between \eqref{UO} and \eqref{UCO} becomes mathematically relevant. For clarity, throughout this work we reserve the term uniform observability for \eqref{UO} and uniform complete observability for \eqref{UCO}.

\section{Problem Formulation}
 \subsection{System dynamics}   
Let $\{\mathcal{I}\}$ denote an inertial frame and $\{\mathcal{B}\}$ a body frame rigidly attached to the IMU. Let $R \in SO(3)$ be the rotation from frame $\{\mathcal{B}\}$ to frame $\{\mathcal{I}\}$, $p \in \mathbb{R}^3$ the position of the origin of $\{\mathcal{B}\}$, expressed in $\{\mathcal{I}\}$, and $v \in \mathbb{R}^3$ the linear velocity of the origin of $\{\mathcal{B}\}$  expressed in $\{\mathcal{B}\}$. 
The dynamics of a rigid body navigating in 3D space are described by
\begin{subequations}\label{eq:system_dynamics}
\begin{align}
    \dot{R} &= R\,\omega^\times \label{eq:system_dynamic_R} \\
    \dot{p} &= Rv \label{eq:system_dynamic_p}  \\
    \dot{v} &= -\omega^\times v + R^\top g + a, \label{eq:system_dynamic_v}  
\end{align}
\end{subequations}
where $\omega \in \mathbb{R}^3$ is the body-frame angular velocity of the rigid body, $ a \in \mathbb{R}^3$ is the apparent acceleration capturing all non-gravitational forces applied to the rigid body expressed in $\{\mathcal{B}\}$, and $g \in \mathbb{R}^3$ is the gravity vector in $\{\mathcal{I\}}$.

 \subsection{Monocular bearing measurements}
 Consider a set of $N$ static landmarks in the environment, denoted by $\mathcal{L} := \{ p_1,\ldots, p_N \}$, where $p_i \in \mathbb{R}^3$ is the constant inertial-frame position of the $i$-th landmark. The relative position of the $i$-th landmark with respect to the origin of  $\{\mathcal{B}\}$, expressed in $\{\mathcal{B}\}$, is given by 
 \begin{equation}\label{eq:Landmark_position_in_body_frame}
     y_i(t)= R^\top ( p_i -p(t)).
 \end{equation}

 Let the pair $(R_c,p_c)$ denote the constant pose of the camera frame $\{\mathcal{C}\}$ with respect to the body frame $\{\mathcal{B}\}$. Then, the relative position of the $i$-th landmark with respect to the origin of $\{\mathcal{C}\}$ expressed in $\{\mathcal{C}\}$ is given by
\begin{equation}\label{eq:Landmark_position_in_the_camera_frame}
    {}^\mathcal{C}y_i(t) = R_c^\top(y_i(t) -p_c).
\end{equation}
The monocular bearing measurements $b_i\in \mathbb{S}^2, i \in \mathbb{I}$, $\mathbb{I} : =\{{1,2, \dots, N}\}$, associated with the $i$-th landmark, are given by
\begin{equation}\label{eq:Landmark_bearing_vector}
    b_i (t)= \frac{{}^ \mathcal{C}y_i(t)}{\|{}^ \mathcal{C}y_i(t) \|} 
\end{equation}
\begin{remark}
    Let $z_i = [u_i, v_i, 1]^\top$ denote the measured homogeneous image coordinates of the $i$-th landmark. The corresponding bearing vector in the camera frame is computed as follows:
    \begin{equation}\label{eq:Landmark_bearing_obtained_from_camera_2D_coordinates}
        b_i = \frac{\mathcal{K}^{-1} z_i}{\|\mathcal{K}^{-1}z_i\|} \in \mathbb{S}^2,
    \end{equation}
where $ \mathcal{K}$ is the camera intrinsic matrix.
\end{remark}

\section{Monocular bearings based VIO}\label{section:mutli-bearings}
\subsection{ Global estimation of the linear velocity, gravity and landmark positions in the body-frame}
In this subsection, we recall the linear time-varying (LTV) observer, proposed in \cite{wangACC25}, that estimates the body-frame landmark positions $y_i,i \in \mathbb{I}$, the body-frame linear velocity $v$, and the body-frame gravity vector $\eta:=R^Tg$, with uniform global exponential stability guarantees. The observer is driven by inertial and visual measurements provided by the IMU and the camera, respectively. In \cite{wangACC25}, the rigid body pose is recovered by using the LTV observer estimates along with the positions, in the inertial frame, of at least three measured landmarks in the body frame. In this work, we use this LTV observer in the VIO context where we assume that the positions of the landmarks in the inertial frame are not available.

From \eqref{eq:system_dynamics}  and differentiating \eqref{eq:Landmark_position_in_body_frame} with respect to time, one obtains
 \begin{equation}\label{eq:kinematic_model_LTV_ multiple_bearing}\begin{cases}
\begin{aligned}
     \dot{v} &= -\omega^\times v + \eta + a \\
    \dot{\eta} & = -\omega^\times \eta\\
    \dot{y}_i &= -\omega^\times y_i - v,\\
\end{aligned}\end{cases}
 \end{equation}
From \eqref{eq:Landmark_bearing_vector}, we introduce the modified outputs for each bearing measurement $b_i, i\in \mathbb{I}$ as
\begin{equation}\label{eq:modified_output}
    d_i := \Pi_i p_c = \Pi_i y_i,
\end{equation}
where $\Pi_i:= \pi(R_c b_i)$ and we have made use of $\pi(R_c b_i)= R_c \pi(b_i)R_c^\top$ and $\pi(R_cb_i)(y_i -p_c) = 0$, for all $i\in\mathbb{I}$. Let the modified output measurements vector be 
$\mathbf{y} := [d_1^\top,d_2^\top,\dots,d_N^\top]^\top \in \mathbb{R}^{3N}$ where $d_i, i\in \mathbb{I}$ is defined in \eqref{eq:modified_output}. Define the state $\mathbf{x}=[v^ \top , \eta ^\top, y_1 ^ \top, \dots, y_N ^ \top] ^ \top \in \mathbb{R}^{(6+3N)}$. 
The resulting LTV system can then be written as
\begin{equation}\begin{cases}
 \begin{aligned}\label{eq:LTV_multi_bearing}
     \dot{\mathbf{x}} &= A\mathbf{x}+ Ba\\
     \mathbf{y} &= C\mathbf{x},
 \end{aligned}\end{cases}
\end{equation}
with the matrices $A(t) \in \mathbb{R} ^{ (6+3N) \times  (6+3N)}$, $B \in \mathbb{R}^{(6+3N) \times 3} $ and $C  \in \mathbb{R} ^ {3N \times (6+3N)}$  defined as

\begin{equation}\label{eq:ABC_matrices_multiple_bearings}
\begin{aligned}
A(t) &= \begin{bmatrix}
-\omega^\times & I_3 & 0_{3\times 3N}\\
0_{3\times 3} & -\omega^\times & 0_{3\times 3N}\\
-\mathbf{1}_N\otimes I_3 & 0_{3N\times 3} & -I_N\otimes \omega^\times
\end{bmatrix},\\[2pt]
B &= \begin{bmatrix} I_3 \\ 0_{3(N+1)\times 3} \end{bmatrix},
C(t) = \begin{bmatrix} 0_{3N\times 6} & \Pi_b \end{bmatrix},
\end{aligned}
\end{equation}

where $\Pi_b = \blkdiag(\Pi_1,\Pi_2,\dots,\Pi_N) \in \mathbb{R}^{3N \times 3N}$.
Now we proceed with the observer design  for the system in \eqref{eq:kinematic_model_LTV_ multiple_bearing} to estimate the body-frame velocity $v$, the body-frame gravity vector $\eta$, and the landmark position $y_i, i \in \mathbb{I}$.

\begin{remark}
The body-frame gravity vector $\eta  = R^ \top g$ is introduced as an additional state to remove the explicit dependence of the linear velocity and the landmark positions dynamics on the attitude $R$, yielding a linear time-varying system.
\end{remark}
Let $\mathbf{\hat{x}} : =  [\hat{v}^ \top, \hat{\eta}^\top, \hat{y}_1^ \top, \ldots, \hat{y}_N^ \top ] ^ \top $ denote the estimate of the state $\mathbf{x}$. The state observer is given by
\begin{equation}\label{eq:Riccati_observer_multiple_bearings}
    \mathbf{\dot{\hat{x}}} = A(t) \mathbf{\hat{x}} + Ba +  K(\mathbf{y} - C(t) \mathbf{\hat{x}}),
\end{equation}
where $\mathbf{\hat{x}}(0) \in  \mathbb{R}^ {(6+3N)}$. The gain matrix $K(t)$  is defined as 
\begin{equation}\label{eq:gain_K_multiple_bearings}
    K(t) = P(t)C(t)^ \top Q^{-1}(t),
\end{equation}
where  $P(t)=P^\top (t)>0$ is the solution of the Riccati equation \eqref{eq:Riccati_equation}, with $P(0)\in \mathbb{R}^ {(6+3N) \times (6+3N)} $, $Q(t)  \in \mathbb{R}^ {3N\times 3N} $ and $V(t) \in \mathbb{R}^ {(6+3N) \times (6+3N)}$ satisfying the conditions specified after \eqref{eq:Riccati_equation}.
 
 Let $\tilde{\mathbf{x}} := \mathbf{x} - \hat{\mathbf{x}} = [\tilde{v}^ \top , \tilde{\eta}^ \top , \tilde{y}_1^ \top, \dots , \tilde{y}_N ^ \top] ^ \top$ denote the state estimation error, with $\tilde{v} : = v - \hat{v}$,  $\tilde{\eta} : = \eta - \hat{\eta}$ and $\tilde{{y}}_i= y_i - \hat{y}_i $, $i \in \{1,2,\dots, N\}$.
 
From \eqref{eq:kinematic_model_LTV_ multiple_bearing} and \eqref{eq:Riccati_observer_multiple_bearings}, the estimation error dynamics are given by 
\begin{equation}\label{eq:closed_loop_system_multiple_bearing}
    \mathbf{\dot{\tilde{x}}} = (A(t)-K(t)C(t))\mathbf{\tilde{x}}
\end{equation}

As shown in \cite[Lemma 1]{wangACC25}, the pair $(A(t),C(t))$ defined in \eqref{eq:ABC_matrices_multiple_bearings} is uniformly observable if there exist constants $\delta, \mu>0$ such that 
    \begin{equation}\label{PE_prop1}
        \int _t ^ {t+\delta} \pi(R(\tau) R_c b_i(\tau))d\tau > \mu I_3
    \end{equation}
 for all $t \geq 0$, $i \in\mathbb{I}$.
Furthermore, assuming that $\omega(t)$ is continuous and uniformly bounded, together with the fact that  the pair $(A(t),C(t))$ is uniformly observable, it follows that the conditions in Lemma \ref{lemma2} are satisfied. Consequently, one can show that the equilibrium point $\tilde{x}=0$ of \eqref{eq:closed_loop_system_multiple_bearing} is uniformly globally exponentially stable (UGES). The proof relies on the existence of the solution $P$ to the Riccati equation \eqref{eq:Riccati_equation} with lower and upper bounds \eqref{Riccati_bounds}. The time-derivative of the Lyapunov function $\mathcal{L}=\tilde{x}^\top P^{-1}\tilde{x}$, along the trajectories of \eqref{eq:closed_loop_system_multiple_bearing}, in view of the Riccati equation \eqref{eq:Riccati_equation}, is given by $\dot{\mathcal{L}}\leq -\frac{v_m}{p_M}\mathcal{L}$. 

The LTV observer in \eqref{eq:Riccati_observer_multiple_bearings} takes the following explicit form: 
\begin{equation}\label{eq:explicit_from}\begin{cases}
\begin{aligned}
     \dot{\hat{v}} &= - \omega^\times \hat{v} + \hat{\eta} + a + K_v (\mathbf{y}- C \mathbf{\hat{x}}) \\
    \dot{\hat{\eta}} & = - \omega^\times \hat{\eta} + K_{\eta}(\mathbf{y}-C\mathbf{\hat{x}})\\
    \dot{\hat{y}}_i &= - \omega^\times \hat{y}_i - \hat{v} + K_{{y}_i}(\mathbf{y}-C\mathbf{\hat{x}}),~~i=1,\hdots,N\\
\end{aligned}\end{cases}
 \end{equation}
with $K = [ K_v^\top,K_{\eta}^\top, K_{{y}_1}^\top, \dots, K_{{y}_N}^\top] ^ \top $ and $K_v,K_{\eta} , K_{{y}_i} \in \mathbb{R} ^{3 \times 3N}$.

\subsection{Pose estimation}\label{section_pose_estimation}
We propose the following pose estimation scheme:
\begin{equation}\label{pose_est1}
\begin{aligned}
\dot{\hat{R}}&=\hat{R}(\omega+\sigma_R)^\times\\
\dot{\hat{p}}&=\hat{R}\hat{v}+(\hat{R}\sigma_R)^\times\hat{p},
\end{aligned}
\end{equation}
where $\sigma_R=k_R (\hat{\eta} \times \hat{R}^\top g)$.
Differentiating $\tilde{R}:=R\hat{R}^\top$ with respect to time, and using \eqref{pose_est1}, one obtains the following rotational error dynamics:
\begin{equation}\label{eq:rotational_error_dynamics1}
    \dot{\tilde{R}}  =  \tilde{R}\left(- k_R (\tilde{R}^ \top g)^\times g  - \Lambda \tilde{\mathbf{x}}\right)^\times
\end{equation}
where $\Lambda$ is given by
\begin{equation}\label{eq: gamma1}
    \Lambda = \begin{bmatrix} {0}_{3 \times 3} & k_R g^\times  \hat{R}& {0}_{3 \times 3N} \end{bmatrix}.
\end{equation}
As per the observer design in the previous subsection, $\tilde{\mathbf{x}}$ is bounded and converges exponentially to zero when $t$ tends to infinity. Therefore, when $t$ tends to infinity, the rotational error dynamic reduces to
\begin{equation}\label{eq:error_R}
    \dot{\tilde{R}}  =  \tilde{R}(- k_R (\tilde{R}^ \top g)^\times g)^\times.
\end{equation}
As shown in \cite{NL_observer_Mouaad}, the equilibrium  $\tilde{R}^\top g=g$ is almost globally asymptotically stable\footnote{An equilibrium point is almost globally asymptotically stable if it is stable and attractive from all initial conditions except a set of zero Lebesgue measure.} for system \eqref{eq:error_R}. Note that  the fact that $\tilde{R}^\top g$ converges to $g$ implies that the pitch and roll estimation errors converge to zero when $t$ tends to infinity.\\
It is worth pointing out that if an additional vector measurement is available, \textit{e.g.,} from a magnetometer, then a full recovery of the orientation is possible.\\
As for the position error dynamics, defining $\tilde{p}=p-\tilde{R}\hat{p}$, one has $\dot{\tilde{p}}=R\tilde{v}$, which shows that $\tilde{p}(t)$ tends to a constant as $t$ tends to infinity in view of the exponential convergence of $\tilde{v}$ to zero.

\section{Monocular Mixed bearings VIO (MBVIO)} \label{MBVIO}
\subsection{Global reduced-order observer for body-frame linear velocity and gravity estimation}

 In this section, a reduced-order observer is developed to estimate the body-frame linear velocity $v$ and the body-frame gravity vector $\eta$. In contrast to the framework introduced in section \ref{section:mutli-bearings}, the landmark positions are not included in the observer state. Instead, the bearing measurements and their time derivatives are used to construct a measurable output relation that provides information about the body-frame velocity. Consequently, the observer state has a fixed dimension that is independent of the number of tracked landmarks. This formulation reduces the computational complexity associated with the landmark position estimation while retaining the visual information required for state correction.

The dynamics of the linear velocity and the gravity vector in the body frame are given by
 \begin{equation}\label{eq:kinematic_model_LTV_ mixed_bearing}\begin{cases}
\begin{aligned}
     \dot{v} &= -\omega^\times v + \eta + a \\
    \dot{\eta} & = -\omega^\times \eta.\\
\end{aligned}\end{cases}
 \end{equation}

Differentiating \eqref{eq:Landmark_bearing_vector} with respect to time, one gets

\begin{equation}\label{eq:bearing_derivative}
    \dot{b}_i = -{\omega_c}^\times b_i - \frac{\pi(b_i)}{l_i} v_c,
\end{equation}
where
\begin{equation}\label{eq:w_c_and_v_c}
    \omega_c = R_c^ \top \omega, \quad v_c = R_c ^ \top (v + \omega^\times p_c),
\end{equation}
and  $l_i := \|{}^ \mathcal{C}y_i\|, i\in \mathbb{I}$ denotes the distance between the camera and the $i$-th landmark, ${}^ \mathcal{C}y_i$ defined in \eqref{eq:Landmark_position_in_the_camera_frame} and we have made use of $\frac{d}{dt}(x/\|x\|) = \frac{1}{\|x\|}\pi(x/\|x\|) \dot{x}$, for $x\not=0$. 

For each $i \in \mathbb{I} $, define
\begin{equation}\label{eq:qi}
    q_i := b_i^\times  (\dot{b}_i + \omega_c^\times b_i)
\end{equation}
It then follows from \eqref{eq:bearing_derivative} that
\begin{equation}\label{q_i}
    q_i = - \frac{1}{l_i} {b_i}^\times v_c,
\end{equation}
and consequently $q_i ^ \top v_c = 0$, from which one has $Mv_c=0$, with 
\begin{equation}\label{eq:mstrix_M}
    M := \sum _{i =1} ^ N \gamma_i q_i q_i ^\top,
\end{equation}
where $\gamma_i>0$.
Using the expression of $v_c$ in \eqref{eq:w_c_and_v_c} yields
\begin{equation}\label{eq:M_barv}
    \bar{M}v = - \bar{M}\omega ^ \times p_c,
\end{equation}
where $\bar{M} = R_c M R_c ^\top$. Considering the virtual output as $\mathbf{y}_r := -\bar{M} \omega^ \times p_c$, from \eqref{eq:M_barv}, one has $\mathbf{y}_r= \begin{bmatrix}   \bar{M} & 0_{3\times 3} \end{bmatrix} \mathbf{x}_r$, with $\mathbf{x}_r=[v^ \top , \eta ^\top] ^ \top \in \mathbb{R}^6$. 
The resulting LTV system is given as
\begin{equation}\begin{cases}
 \begin{aligned}\label{eq:LTV_multi_bearing}
     \dot{\mathbf{x}}_r &= A (t)\mathbf{x}_r+ Ba\\
     \mathbf{y}_r &= C(t)\mathbf{x}_r,
 \end{aligned}\end{cases}
\end{equation}

with the matrices $A(t) \in \mathbb{R} ^{ 6 \times 6}$, $B \in \mathbb{R}^{6 \times 3} $ and $C(t)  \in \mathbb{R} ^ {3\times 6}$  defined as

\begin{equation}\label{eq:ABC_matrices_mixed_bearings}
\begin{aligned}
A(t) &= \begin{bmatrix}
-\omega^\times & I_3 \\
0_{3\times 3} & -\omega^\times \\
\end{bmatrix},\\[2pt]
B &= \begin{bmatrix} I_3 \\ 0_{3\times 3} \end{bmatrix},
C(t) = \begin{bmatrix}   \bar{M} & 0_{3\times 3} \end{bmatrix}.
\end{aligned}
\end{equation}
Let $\hat{\mathbf{x}}_r := [\hat{v}^\top, \hat{\eta}^\top]^\top $ be the state estimate of $\mathbf{x}_r$. We propose the following observer:

\begin{equation}\label{eq:riccati_observer_mixed_bearings}
    \dot{\hat{\mathbf{x}}}_r = A(t)\hat{\mathbf{x}}_r + Ba + K_r(t) ( {\mathbf{y}}_r -C(t)\hat{\mathbf{x}}_r),
\end{equation}
where $\hat{\mathbf{x}}_r (0) \in \mathbb{R}^6$. The gain matrix $K_r(t)$ is defined as 
\begin{equation}\label{eq:Klaman_gain_mixed_bearings}
    K_r(t) = P(t) C(t)^\top Q^{-1}(t)
\end{equation}
where $P(t)=P^\top(t)>0$ is the solution of the Riccati equation \eqref{eq:Riccati_equation}, with $A(t)$, $B$, and $C(t)$ defined in \eqref{eq:ABC_matrices_mixed_bearings} and $ P(0) \in \mathbb{R}^{6\times6}$ being a symmetric positive definite matrix. The matrices $Q(t) \in \mathbb{R}^{3\times3}$ and $V(t) \in \mathbb{R}^{6\times6}$ satisfy the conditions provided after \eqref{eq:Riccati_equation}.\\
Let $\tilde{\mathbf{x}}_r := \mathbf{x}_r - \hat{\mathbf{x}}_r = [\tilde{v}^ \top , \tilde{\eta}^ \top ] ^ \top$ be the state estimation error, with $\tilde{v} : = v - \hat{v}$ and $\tilde{\eta} : = \eta - \hat{\eta}$.

From \eqref{eq:LTV_multi_bearing} and \eqref{eq:riccati_observer_mixed_bearings}, the estimation error dynamics are given by
\begin{equation}\label{eq:close_loop_sys_mixed_bearings}
    \dot{\tilde{\mathbf{x}}}_r = (A(t) -K_r(t) C(t))\tilde{\mathbf{x}}_r
\end{equation}
where $(A(t),C(t))$ are defined in \eqref{eq:ABC_matrices_mixed_bearings}.
The following proposition provides some sufficient conditions to satisfy the conditions in Lemma \ref{lemma2}:
 \begin{proposition} \label{Proposition1}
    Assume that $\omega(t)$ and $\dot{b}_i(t)$, $i=1,\hdots, N$, are continuous and uniformly bounded. Assume that there exist constants $\delta_M, \mu_M >0$ such that for all $t\geq0$
     \begin{equation}\label{PE_cond}
         \sum _{i=1} ^ N \int _t ^ {t+\delta_M} \gamma_i \bar{q}_i(\tau) \bar{q}_i ^\top (\tau)  d\tau \geq \mu_MI_3,
     \end{equation}
     where $\bar{q}_i(t):= R(t)R_cq_i(t)$, and $q_i$ is defined in \eqref{q_i}.
     Then, the conditions of Lemma \ref{lemma2} are satisfied.
 \end{proposition}
 
 \noindent\textit{Proof.} See Appendix~\ref{proof_uni_obs}.

 It is clear that condition \eqref{PE_cond} is satisfied as long as there exists at least one vector $\bar{q}_i$  that is persistently exciting over every time window of length $\delta_M$, \textit{i.e.,} satisfying \eqref{PE_cond}. The following proposition provides a geometric interpretation of persistency of excitation condition \eqref{PE_cond}.
 \begin{proposition}\label{Prop_suff_cond_excit}
Sufficient conditions guaranteeing the satisfaction of \eqref{PE_cond} are as follows:
 \begin{itemize}
 \item For every time $t$ such that $v_c(t)\not=0$, there is no direction in the plane orthogonal to $v_c$, which is orthogonal to all bearing $b_i(t)$, $i \in \mathbb{I}$, and
 \item The vector $R(t)R_c v_c$ does not remain aligned with a fixed direction throughout any interval of length $\delta_M$.
 \end{itemize}
 \end{proposition}
 
 \noindent\textit{Proof.} See Appendix~\ref{proof_prop_suff_cond}.

The following proposition establishes the stability properties of the equilibrium point $\tilde{\mathbf{x}}_r=0$ of \eqref{eq:close_loop_sys_mixed_bearings}:
\begin{proposition}
Under the conditions of Proposition \ref{Proposition1}, the equilibrium point $\tilde{\mathbf{x}}_r=0$ of \eqref{eq:close_loop_sys_mixed_bearings} is UGES.
\end{proposition}
\begin{proof}
The detailed proof is omitted as it is straightforward using Lyapunov function $\mathcal{L}=\tilde{\mathbf{x}}_r^\top P^{-1}\tilde{\mathbf{x}}_r$, under the conditions of Proposition \ref{Proposition1}.
\end{proof}

The explicit form of the LTV observer in \eqref{eq:riccati_observer_mixed_bearings} is given as
\begin{equation}\label{eq:explicit_from_mixed_bearing}\begin{cases}
\begin{aligned}
     \dot{\hat{v}} &= - \omega^\times \hat{v} + \hat{\eta} + a + K_{rv} (\mathbf{y}_r- C (t) {\hat{\mathbf{x}}}_r) \\
    \dot{\hat{\eta}} & = - \omega^\times \hat{\eta} + K_{r\eta}(\mathbf{y}_r-C(t){\hat{\mathbf{x}}}_r),\\
\end{aligned}\end{cases}
 \end{equation}
with $K_r = [ K_{rv}^\top,K_{r\eta}^\top]^\top $ and $K_{rv},K_{r\eta}  \in \mathbb{R} ^{3\times3}$.

The position and orientation are estimated as in section \ref{section_pose_estimation}.

\section{Monocular mixed bearings VIO with IMU-bias estimation}\label{MBVIO_bias_est}
\subsection{Reduced-order observer for body-frame linear velocity, gravity and IMU-bias estimation}
In this section, the reduced mixed-bearings observer presented in the previous section is extended to account for the gyroscope and accelerometer biases. Let $\omega_m$ and $a_m$ denote gyroscope and accelerometer measurements and let $b_g$ and $b_a$ denote the constant gyroscope and accelerometer biases such that
\begin{equation}\label{eq:measurement_model_IMU}
\begin{aligned}
     \omega_m = \omega + b_g\\
     a_m = a + b_a,
\end{aligned}
\end{equation}
Define the augmented state vector as $\mathbf{x}_b:= [v^\top, \eta^\top, b_a ^\top , b_g^\top]^\top$, the measured input as $u_m = [\omega_m ^\top , a_m ^ \top]^\top $.
System \eqref{eq:kinematic_model_LTV_ mixed_bearing} augmented by the constant biases dynamics is given by 
\begin{equation}\label{eq:augmented_dynamics_compact}
    \dot{\mathbf{x}}_b = f_b(\mathbf{x}_b, u_m),
\end{equation}
where
\begin{equation}\label{eq:augmented_dynamics}
    f_b(\mathbf{x}_b,u_m) = 
    \begin{bmatrix}
         -(\omega_m - b_g)^\times v + \eta + a_m - b_a\\
         -(\omega_m - b_g)^\times \eta\\
          0_{3\times1}\\
          0_{3\times1}
    \end{bmatrix},
\end{equation}
Let $\hat{\mathbf{x}}_b:= [\hat{v}^\top, \hat{\eta}^\top, \hat{b}_a ^\top , \hat{b}_g^\top]^\top \in \mathbb{R}^{12}$ denote the estimate of $\mathbf{x}_b$, and define the bias-compensated angular velocity estimate as $\hat{\omega}:= \omega_m - \hat{b}_g$.
One can show that 
\begin{equation}
f_b(\mathbf{x}_b, u_m)-f_b(\hat{\mathbf{x}}_b, u_m)=A(t)\mathbf{\tilde{x}}_b+\phi({\mathbf{\tilde{x}}}_b),
\end{equation}
where the matrices $A(t)\in \mathbb{R}^{12\times12}$ and $\phi({\mathbf{\tilde{x}}}_b)\in \mathbb{R}^{12}$ are given by
\begin{equation}\label{eq:A_phi_matrices_reduced_order_and_bias}
    \begin{aligned}
      A(t) =  \begin{bmatrix}
        -\hat{\omega}^\times & I_3 & -I_3& - \hat{v}^\times\\
        0_{3\times3} & -\hat{\omega}^\times & 0_{3\times3}& -\hat{\eta}^\times \\
        0_{3\times3}&0_{3\times3}&0_{3\times3}&0_{3\times3}\\
        0_{3\times3}&0_{3\times3}&0_{3\times3}&0_{3\times3}
     \end{bmatrix} , ~
     \phi(\tilde{\mathbf{x}}_b) =\begin{bmatrix}
      \tilde{b}_g ^\times \tilde{v} \\
      \tilde{b}_g^\times \tilde{\eta}\\
      0_{3\times1}\\
      0_{3\times1}
  \end{bmatrix} 
    \end{aligned}
\end{equation}
and $\tilde{\mathbf{x}}_b := \mathbf{x}_b - \hat{\mathbf{x}}_b = [\tilde{v}^ \top , \tilde{\eta}^ \top , \tilde{b}_a^ \top, \tilde{b}_g ^ \top] ^ \top$ denotes the state estimation error, with $\tilde{v} : = v - \hat{v}$,  $\tilde{\eta} : = \eta - \hat{\eta}$, $\tilde{b}_a : = b_a - \hat{b}_a$, and $\tilde{b}_g : = b_g - \hat{b}_g$.

Let $\hat{\omega}_c:= R_c^ \top \hat{\omega}$ and $\hat{v}_c := R_c ^\top (\hat{v} + \hat{\omega} ^ \times p_c)$. Using $q_i$ defined in \eqref{eq:qi} and $v_c$ defined in \eqref{eq:w_c_and_v_c}, one has $q_i ^\top v_c=0$, and hence, one can show that 
\begin{equation}\label{output}
    q_i ^\top v_c= \hat{q}_i^\top \hat{v}_c + C_i \tilde{\mathbf{x}}_b + \bar{\phi}_i (\tilde{\mathbf{x}}_b)=0,
\end{equation}
where, for each $i \in \{1,2,\dots, N\}$,  $\hat{q}_i = b_i ^\times ( \dot{b}_i + \hat{\omega}_c ^ \times b_i)$ , $C_i = [\hat{q}_i ^\top R_c ^\top, 0_{1\times3} , 0_{1\times3}, \alpha_i]$, with $\alpha_i = \hat{q}_i ^\top R_c ^\top p_c ^ \times - \hat{v}_c^\top \pi(b_i) R_c ^\top$, and $\bar{\phi}_i (\tilde{\mathbf{x}}_b) = -\tilde{b}_g ^\top R_c \pi(b_i)  R_c ^\top \tilde{v} - \tilde{b}_g ^\top R_c \pi(b_i) R_c ^\top p_c^\times \tilde{b}_g$.\\
Let $\bar{\phi}(\tilde{\mathbf{x}}_b) = [\bar{\phi}_1(\tilde{\mathbf{x}}_b)\dots,\bar{\phi}_N(\tilde{\mathbf{x}}_b)]^\top \in \mathbb{R}^N$, $\hat{\mathcal{Q}}= [\hat{q}_1,\dots,\hat{q}_N]^\top \in \mathbb{R}^{N\times3}$ and $C(t) \in \mathbb{R}^{N\times12} $ be defined as 
\begin{equation}\label{c_matrix}
    C(t) : = 
    \begin{bmatrix}  C_1^\top & \cdots  & C_N^\top \end{bmatrix} ^ \top.
\end{equation}
From \eqref{output}, it follows that
\begin{equation}\label{eq:matrix_C_reduced_order_bias}
    \hat{\mathcal{Q}}\hat{v}_c + C\tilde{\mathbf{x}}_b + \bar{\phi}(\tilde{\mathbf{x}}_b ) =0.
\end{equation}
We propose the following state observer: 
\begin{equation}\label{eq:augmented_dynamics_estimate}
    \dot{\hat{\mathbf{x}}}_b = f_b(\hat{\mathbf{x}}_b, u_m) - K_b\hat{\mathcal{Q}}\hat{v}_c,
\end{equation}
leading to the following estimation error dynamics:
\begin{equation}\label{eq:closed-loop_sys_reduced_order_bias}
    \dot{\tilde{\mathbf{x}}}_b =
    (A(t) - K_b(t) C(t)) \tilde{\mathbf{x}}_b + \phi(\tilde{\mathbf{x}}_b) - K_b(t) \bar{\phi}(\tilde{\mathbf{x}}_b), 
\end{equation}
where the gain matrix $K_b(t)$ is defined as $K_b(t) = P(t) C(t)^\top Q^{-1}(t)$, the matrix $P(t) = P^\top (t) >0$ is the solution of the continuous Riccati  equation \eqref{eq:Riccati_equation}, with $P(0) \in \mathbb{R}^{12\times12}$ being a symmetric positive definite matrix, and the matrices $V(t) \in \mathbb{R}^{12\times12}$ and $Q(t) \in \mathbb{R}^{N\times N}$ satisfy the conditions given after \eqref{eq:Riccati_equation}. The matrices $A(t)$ and $C(t)$ are defined in \eqref{eq:A_phi_matrices_reduced_order_and_bias} and \eqref{c_matrix}, respectively.

Now, one can conclude that if the conditions of Lemma \ref{lemma1} hold, then the equilibrium point $\tilde{\mathbf{x}}_b=0$ is locally exponentially stable for the closed loop system \eqref{eq:closed-loop_sys_reduced_order_bias}. The detailed proof is omitted here as it follows directly from the Lyapunov  Lyapunov function $\mathcal{L}=\tilde{\mathbf{x}}_b^\top P^{-1}\tilde{\mathbf{x}}_b$, the Riccati equation \eqref{eq:Riccati_equation} and the fact that $\lim_{||\tilde{\mathbf{x}}_b||\rightarrow 0} \frac{||\phi(\tilde{\mathbf{x}}_b)||}{||\tilde{\mathbf{x}}_b||}=0$ and $\lim_{||\tilde{\mathbf{x}}_b||\rightarrow 0} \frac{||\bar{\phi}(\tilde{\mathbf{x}}_b)||}{||\tilde{\mathbf{x}}_b||}=0$.

 The explicit form of the observer in \eqref{eq:augmented_dynamics_estimate} is the given by
 \begin{equation}
     \begin{cases}
         \begin{aligned}
             \dot{\hat{v}} &= -\hat{\omega}^\times \hat{v}+\hat{\eta} + a_m -\hat{b}_a- K_b^v \hat{\mathcal{Q}}\hat{v}_c,\\
             \dot{\hat{\eta}} &= -\hat{\omega}^\times \hat{\eta}- K_b^\eta \hat{\mathcal{Q}}\hat{v}_c, \\
             \dot{\hat{b}}_a &= - K_{b}^a \hat{\mathcal{Q}}\hat{v}_c,\\
             \dot{\hat{b}}_g & = -K_{b}^g \hat{\mathcal{Q}}\hat{v}_c,
         \end{aligned}
     \end{cases}
 \end{equation}
with $K_b = [{K_{b}^v}^\top,{K_{b}^\eta}^\top,{K_{b}^a}^\top, {K_{b}^g }^\top]^\top \in \mathbb{R}^{12\times N}$.

    
The position and orientation are estimated as in section \ref{section_pose_estimation} using the bias-compensated angular velocity estimate $\hat{\omega} = \omega_m - \hat{b}_g$.

\section{Simulation results}

In this section, we evaluate and compare the performance of the three proposed VIO schemes: the monocular bearing-based VIO (MVIO), the monocular mixed-bearing VIO (MBVIO), and the monocular mixed-bearing VIO with IMU bias estimation (MBVIO-B), using a common simulated trajectory. 
Although the observer design and convergence analysis are developed in continuous time under ideal, noise-free measurements, the numerical implementation uses sampled and noisy IMU and camera data. The observer states and Riccati matrix are propagated numerically between camera frames using the IMU measurements, while visual corrections are applied at the camera sampling instants, as illustrated in Algorithm \ref{alg:mono_mixed_observer}. Bearing derivatives are approximated from successive tracked image measurements. The matrices $V$ and $Q$ are selected using the assumed noise characteristics and positive-definite regularization terms as show in subsection \ref{noise}. 
 The simulation has a total duration of $50$~s. IMU measurements are generated at $200$~Hz, while camera measurements are generated at $20$~Hz. The position of the vehicle is given by $p(t) = [ 3 \cos (t), 3 \sin( t), \sin(2 t)]^\top$ and its orientation is given by $R(t) = R_z ( \psi(t))R_y(\theta(t)) R_x(\phi(t))$, where $\phi(t)= 0$, $\theta(t) = - \arctan \bigl( \frac{2}{3} \cos(2 t) \bigr)$, and $\psi(t) =  t+ \frac{\pi}{2}$. The environment is a rectangular room of dimension $10 \times 10 \times 4$~m. A total of $450$ landmarks are distributed within the interior of the room. A monocular pinhole camera model is used to generate the visual measurements, with its optical axis aligned with the forward direction of the body, and its center located at the origin of the body. The camera extrinsic parameters are given by $R_c=R_z(-\pi/2)R_x(-\pi/2)$ and $p_c=0_{3\times 1}$.
We track up to 40 landmarks in each frame. In the next frame, we keep the landmarks that are still visible and add new ones to replace those that are lost, up to a maximum of 40 landmarks per frame. The simulated camera has focal lengths $(f_x,f_y)=(640,623.54)~\mathrm{px}$, principal point $(c_x,c_y)=(640,360)~\mathrm{px}$, and image resolution $1280\times720~\mathrm{px}$. After expressing each landmark in the camera frame as ${}^ \mathcal{C}y_i = [{}^ \mathcal{C}X_i,{}^ \mathcal{C}Y_i,{}^ \mathcal{C}Z_i]^\top$, $i\in \mathbb{I}$, its pixel coordinates are obtained as $[u_i,v_i,1]^\top = \frac{\mathcal{K}}{{}^ \mathcal{C}Z_i}{}^ \mathcal{C}y_i$, with $\mathcal{K}$ being the camera intrinsic matrix. A landmark is considered visible if ${}^ \mathcal{C}Z_i > 0$, $0 \le u_i < 1280$, and $0 \le v_i < 720$.\\
The IMU measurements are corrupted by white noise, with standard deviations $\sigma_\omega = 2.4 \times 10 ^ {-3} ~\mathrm{rad}/ \mathrm{s}$ for the gyroscope and $\sigma_a =2.83 \times 10 ^{-2} \mathrm{m}/ \mathrm{s}^2$ for the accelerometer. 
Image measurements are subject to zero-mean Gaussian noise with a standard deviation $\sigma_{\mathrm{px}} = 0.5 ~ \mathrm{px}$. 

For the MBVIO-B simulation, constant accelerometer and gyroscope biases $b_a = [0.10, -0.08, 0.12]^\top$ and $b_g = [0.005, -0.003, 0.008]^\top$ are added to the simulated IMU measurements.
For the three estimation schemes, we set $\hat{R}(0) = R_z (\pi/6)R_y(\pi/6) R_x(\pi/6)$, $\hat{v}(0)= \hat{p} (0)= 0_{3\times1}$, $\hat{\eta}(0) = \hat{R}^\top(0)g$. For the mixed-bearing VIO with IMU bias estimation, we set $\hat{b}_a (0) = \hat{b}_g (0) = 0_{3\times1}$. The gain of the attitude correcting term is set to $k_R = 1$. For the MVIO, we set $\hat{y}_i(0)=0$ and for the MBVIO, we set $\gamma_i =1$, for each landmark $i$.
The matrices $Q$ and $V$ involved in the Riccati equation are selected according to the guidelines provided in the next subsection.  For MVIO, MBVIO, and MBVIO-B, we set $P(0) = 10^4 I_{6+3N}$, $P(0) = I_{6}$, and $P(0) = I_{12}$, respectively.
The estimated position trajectory $\hat{p}$ is aligned with the ground-truth position trajectory $p$ using the Umeyama method \cite{Umeyama}, yielding aligned position $\hat{p}_{\text{aligned}} = R_u \hat{p}+ p_u$, where $(R_u, p_u) \in \mathrm{SE}(3)$ denotes the optimal alignment transformation.
\\
Computing $q_i$ in \eqref{eq:qi} requires the time derivative of the bearing vector. This derivative was approximated using endpoint and mid-point polynomial differentiation based on the Savitzky–Golay formulation \cite{SGolay}. For the MBVIO scheme, a quadratic polynomial was fitted component-wise to the bearing measurements of landmark $i$ over three consecutive frames, at times ${t_{k-2},t_{k-1},t_k}$. Evaluating the derivative of the fitted polynomial at the end point $t_k$ yields
\begin{equation}{\label{eq:three_frame_derivative}}
    \dot{b}_{i,k} \approx \frac{3b_{i,k} - 4b_{i,k-1} + b_{i,k-2}}{ 2 \Delta t}
\end{equation}
 where $\Delta t = t_k - t_{k-1}$ is the constant camera sampling time. This approximation yielded better simulation results than the first-order backward-difference approximation.
 For the MBVIO-B scheme, the three-frame approximation did not provide satisfactory performance. A fourth-degree polynomial was therefore fitted component-wise to the bearing measurements over five consecutive frames, at times ${t_{k-4},t_{k-3},t_{k-2},t_{k-1},t_k}$. The derivative evaluated at the midpoint $t_{k-2}$ is given by
 \begin{equation}{\label{eq:five_frame_derivative}}
     \dot{b}_{i,k-2} \approx \frac{b_{i,k-4} - 8b_{i,k-3} + 8b_{i,k-1} -b_{i,k}}{12 \Delta t}
 \end{equation}
Note that the derivative at $t_{k-2}$ becomes available only upon the arrival of the measurements at $t_k$, introducing a delay of $2 \Delta t$. To ensure temporal consistency, the visual correction is applied to the stored state  $(\hat{R}, \hat{p}, \hat{v}, \hat{\eta}, \hat{b}_a, \hat{b}_g)$ and Riccati solution $P$  at $t_{k-2}$, using the bearing and IMU measurements at that timestamp. The corrected state and Riccati solution are then propagated to $t_k$ using the corresponding IMU measurements over $[t_{k-2}, t_k]$. The resulting MBVIO-B implementation is summarized in Algorithm~\ref{alg:mbviob_delayed}.\\
Simulation results are presented in Figs. \ref{fig:estimation-errors}, \ref{fig:3D-traj}, and \ref{fig:Accel-gyro-bias-error}. One can observe that the velocity, gravity, roll, pitch and IMU-biases estimation errors converge to the vicinity of zero within a few seconds, while the yaw error does not converge to zero due to its inherent unobservability. The aligned estimated position trajectories closely follow the ground-truth trajectory. 

\begin{algorithm}[htbp]
\caption{Monocular Mixed Bearing VIO (MBVIO)}
\label{alg:mono_mixed_observer}
\begin{algorithmic}[1]
\INPUT \,Continuous IMU measurements $(\omega(t), a(t))$; intermittent visual measurements {$z_{i,k}$ }, $i \in \mathbb{I}$ at times $\{t_k\}_{k\in\mathbb{N}_{>0}}$.
\OUTPUT \,Estimates $\hat{R}(t)$, $\hat{p}(t)$, $\hat{v}(t)$, $\hat{\eta}(t)$, $\forall t\ge 0$.
\FOR{$k\ge 1$}
  \WHILE{$t \in [t_{k-1}, t_k)$}
    \STATE $\dot{\hat{R}} = \hat{R}(\omega + \sigma_R)^\times$        
    \STATE $\dot{\hat{p}} = \hat{R}\hat{v}+ (\hat{R}\sigma_R)^\times \hat{p}$
    \STATE $\dot{\hat{v}} = -\omega^\times \hat{v} + \hat{\eta} + a$
    \STATE $\dot{\hat{\eta}} = -\omega^\times \hat{\eta}$
     \STATE $\dot{P} = A(t)P + P A^\top(t) + V(t)$
    
  \ENDWHILE
  \STATE { Compute bearing vectors  $b_{i,k}$ using $z_{i,k}$} 
  \STATE Construct the virtual output $\mathbf{y}_{r,k}$ 
  \STATE $\mathbf{z}_{r,k} = \mathbf{y}_{r,k} - C_k\hat{\mathbf{x}}_{r,k} ^-$   
  \STATE $K_{r,k} = P_k ^- C^\top_k (C_k P_k^- C_k^\top + Q_k)^{-1}$

  \STATE $\sigma_{r,k} = K_{r,k} \mathbf{z}_{r,k}$
  \STATE Extract $\sigma_{v,k},\sigma_{\eta,k}$ from $\sigma_{r,k}$ 
  \STATE $\hat{R}^{+} _k = \hat{R}^-_k$
  \STATE $\hat{p}^{+} _k= \hat{p}^-_k$
  \STATE $\hat{v}^{+}_k= \hat{v}^-_k + \sigma_{v,k}$
  \STATE $\hat{\eta}^{+} _k= \hat{\eta}^-_k + \sigma_{\eta,k}$
  \STATE $P^+ _k= (I_{6} - K_{r,k}C_k) P^-_k$
\ENDFOR
\end{algorithmic}
\end{algorithm}

\begin{algorithm}[htbp]
    \caption{MBVIO-B with delayed visual correction}
    \label{alg:mbviob_delayed}
    \begin{algorithmic}[1]

    {
    \INPUT \,Continuous IMU measurements $\omega_m(t), a_m(t)$; intermittent visual measurements $z_{i,k},~ i \in \mathbb{I}$ at times $\{t_k\}_{k\in\mathbb{N}_{>0}}$
    \OUTPUT \,Estimates $\hat{R}(t)$, $\hat{p}(t)$, $\hat{v}(t)$, $\hat{\eta}$, $\hat{b}_a$, $\hat{b}_g$, $\forall t\ge 0$
    \FOR{$k\ge 1$}
   \WHILE{$t \in [t_{k-1}, t_k)$}
    \STATE $\dot{\hat{R}} = \hat{R}((\omega_m - \hat{b}_g) + \sigma_R)^\times$       
    \STATE $\dot{\hat{p}} = \hat{R}\hat{v} +(\hat{R}\sigma_R)^\times \hat{p}$
    \STATE $\dot{\hat{v}} = -{(\omega_m - \hat{b}_g)}^\times \hat{v} + \hat{\eta} + a_m - \hat{b}_a$
    \STATE $\dot{\hat{\eta}} = -{(\omega_m - \hat{b}_g)}^\times \hat{\eta}$
   \STATE  $\dot{\hat{b}}_a = 0_{3\times 1}$ , \quad $\dot{\hat{b}}_g = 0_{3\times 1}$
   \STATE $\dot{P} = A(t)P + P A^\top(t) + V(t)$ 
    \STATE  Store the state $(\hat{R}, \hat{p}, \hat{v}, \hat{\eta}, \hat{b}_a, \hat{b}_g, P)$ and the IMU measurements $(\omega_m, a_m)$
  \ENDWHILE  
  \STATE Compute and store $b_{i,k}$ using $z_{i,k}$  
    \IF{$k\geq5$} 
     \STATE Select $\mathcal{J}_k$, the set of landmarks observed in all five frames from $t_{k-4}$  to $t_k$
   \IF{$\mathcal{J}_k \neq \varnothing$}    
  \STATE Compute $\dot b_{j,k-2}$ for all $j\in\mathcal{J}_k$ 
        \STATE Retrieve the stored state at $t_{k-2}$ and Construct $\hat{\mathcal{Q}}_{k-2}$, $C_{k-2}$, and $Q_{k-2}$
        \STATE $\mathbf{z}_{k-2}=-\hat{\mathcal{Q}}_{k-2}^\top\hat{v}_{c,k-2}^{-}$
        \STATE \scalebox{0.94}{$K_{k-2}
        =P_{k-2}^{-}C_{k-2}^\top\left(C_{k-2}P_{k-2}^{-}C_{k-2}^\top+Q_{k-2}\right)^{-1}$ }
        \STATE $\sigma_{k-2}=K_{k-2}\mathbf{z}_{k-2}$ 
        \STATE Extract $\sigma_v,\sigma_\eta,\sigma_{b_a},\sigma_{b_g}$ from $\sigma_{k-2}$ 
        \STATE $\begin{aligned}
        \hat{R}_{k-2}^{+} & =\hat{R}_{k-2}^{-}, &\qquad
        \hat{p}_{k-2}^{+} &=\hat{p}_{k-2}^{-}\\
        \hat{v}_{k-2}^{+} &=\hat{v}_{k-2}^{-}+\sigma_v, &\qquad \hat{\eta}_{k-2}^{+} &=\hat{\eta}_{k-2}^{-}+\sigma_\eta \\
        \hat{b}_{a_{k-2}}^{+} &=\hat{b}_{a_{k-2}}^{-}+\sigma_{b_a}, &\qquad \hat{b}_{g_{k-2}}^{+} &=\hat{b}_{g_{k-2}}^{-}+\sigma_{b_g}
        \end{aligned}$
        
        \STATE $P_{k-2}^{+}
        =(I_{12}-K_{k-2}C_{k-2})P_{k-2}^{-}$
    \STATE Repeat Steps~3-10 and update the stored state over  $[t_{k-2},t_k]$
    using the stored IMU measurements    
    \ELSE
         \STATE Return to step 4 
      \ENDIF
\ENDIF
\ENDFOR}
    \end{algorithmic}
\end{algorithm}

\subsection{Interpretation of $Q$ and $V$ in terms of noise characteristics}\label{noise}
    Although the proposed observers are designed within a deterministic framework, the matrices $V(t)$ and $Q(t)$ appearing in the Riccati equation can be given a meaningful interpretation in terms of the noise characteristics of an associated stochastic system. Let the noisy gyroscope and accelerometer measurements be modeled as 
    \begin{equation}
        \omega_m = \omega + n_\omega, \qquad a_m= a+ n_a
    \end{equation}
    where $n_\omega, n_a \in \mathbb{R}^3$ are zero-mean white Gaussian noise signals, and let $n_x := [n_\omega^\top \ n_a^\top]^\top$. Their covariance matrices are given by \begin{equation*}
        \operatorname{Cov}(n_\omega) = \sigma_\omega ^2 I_3, \qquad \operatorname{Cov}(n_a) = \sigma_a ^2 I_3
    \end{equation*}
    such that $ \operatorname{Cov}(n_x) = diag(\sigma_\omega^2 I_3 , \sigma_a^2 I_3)$, where $\sigma_\omega$ and $\sigma_a$ denote the standard deviations of the gyroscope and accelerometer noise, respectively.
    
    For each estimation scheme, the state estimation error dynamics can be written in the form
    \begin{equation}
        \dot{\tilde{x}} = A(t) \tilde{x} + G(t) n_x,
    \end{equation}
    where $A(t)$ is defined in \eqref{eq:ABC_matrices_multiple_bearings} for MVIO, in \eqref{eq:ABC_matrices_mixed_bearings} for MBVIO, and in \eqref{eq:A_phi_matrices_reduced_order_and_bias} for MBVIO-B.  The process noise covariance matrix is accordingly selected as 
    \begin{equation}{\label{eq:time_varying_process_noise_mat_V}}
        V(t) = G(t) \operatorname{Cov}(n_x) G^ \top(t),
    \end{equation}
    where the matrix $G(t)$ is given, for the MVIO, MBVIO, and MBVIO-B, respectively, as follows: 
    \[
    G(t) = \begin{bmatrix}
            -\hat{v}^\times & -I_3 \\
             -\hat{\eta}^\times & 0_3 \\
             -\hat{Y} & 0_{3N \times 3} \\
        \end{bmatrix}
    \]   
    \[
    \begin{aligned}
        G(t) &= \begin{bmatrix}
            -\hat{v}^\times & -I_3 \\
             -\hat{\eta}^\times & 0_3 \\ 
        \end{bmatrix}, \qquad
                G(t) = \begin{bmatrix}
            -\hat{v}^\times & -I_3 \\
             -\hat{\eta}^\times & 0_3 \\
             0_{6\times3} & 0_{6\times3}
        \end{bmatrix}
    \end{aligned}      
    \]
    where $\hat{Y} = \left[(\hat{y}_1^\times) ^\top ~ \dots ~ (\hat{y}_N ^ \times ) ^ \top\right]^\top$. Since $V(t)$ has to be uniformly positive definite, in our simulations, we consider $V(t)=G(t) \operatorname{Cov}(n_x) G^ \top(t)+\epsilon_V I$, with $\epsilon_V= 10^{-6}$.\\
To determine the $Q$ matrix, only the contribution of the gyroscope noise to the output is considered, while the noise affecting the bearing measurements is neglected. Since in the MVIO observer, the output is not affected by the gyro noise, we pick $Q=10^{-4}I$. For the MBVIO and MBVIO-B, we determine $Q(t)$ as follows.\\
For the MBVIO, define  $\omega_{c,m}:= R_c ^\top \omega_m$,  $y_m := -\bar{M}_m \omega_m^ \times p_c$, $q_{i,m} := b_i^{\times} (\dot{b}_i + \omega_{c,m} ^\times b_i)$, $\hat{v}_{c,m}:= R_c^\top(\hat{v} + \omega_m ^\times p_c)$,  $\bar{M}_m = R_c \sum_{i=1}^N \gamma_i q_{i,m}q_{i,m}^\top R_c ^\top$ and $C_m=[\bar{M}_m ~ 0_{3\times 3}]$.\\
 The measurement residual can therefore be written as 
\begin{equation}
    y_m - C_m(t) \hat{x} \approx C(t) \tilde{x} + H(t) n_\omega,
\end{equation}
    \begin{equation}
        H(t) = \bar{M}_m p_c^\times - R_c \left(\sum_{i=1} ^N \gamma_i q_{i,m} \hat{v}_{c,m} ^\top\pi(b_i) \right) R_c^\top.
    \end{equation}
For the MBVIO-B estimation scheme,  define $\hat{\omega}_{c,m}:= R_c^\top(\omega_m - \hat{b} _g)$, $\hat{q}_{i,m} := b_i^{\times} (\dot{b}_i +\hat{\omega}_{c,m} ^\times b_i)$ and $\hat{v}_{c,m}:= R_c^\top(\hat{v} + \hat{\omega}_m ^\times p_c)$.
 Then,  
    \begin{equation}
        H(t) = \begin{bmatrix}
            \hat{q}_{1,m} ^\top R_c^\top p_c^\times - \hat{v}_{c,m}^\top \pi(b_1) R_c ^ \top \\
            \vdots\\
            \hat{q}_{N,m} ^\top R_c^\top p_c^\times - \hat{v}_{c,m}^\top \pi(b_N) R_c ^ \top,
        \end{bmatrix}
    \end{equation}
The measurement noise covariance matrices for the MBVIO and MBVIO-B estimation schemes are obtained as follows:
    \begin{equation}{\label{eq:time_varying_measurements_noise_matrix_Q}}
        Q(t) = H(t) \operatorname{Cov}(n_\omega)H^\top (t).
    \end{equation}  
 Since $Q(t)$ has to be uniformly positive definite, in our simulations, we consider $Q(t) = H(t) \operatorname{Cov}(n_\omega)H^\top (t)+\epsilon_Q I$, with $\epsilon_Q=5 \times10^{-2}$ for MBVIO, and $\epsilon_Q=5 \times10^{-3}$ for MBVIO-B.

\begin{figure*}[htbp]
    \centering
    \subfloat[\label{fig:position-error}]{%
    \includegraphics[width=0.32\textwidth]{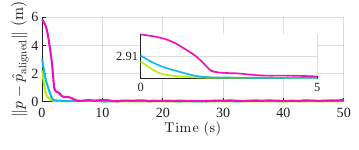}}
    \hfill
    \subfloat[\label{fig:velocity-error}]{%
    \includegraphics[width=0.32\textwidth]{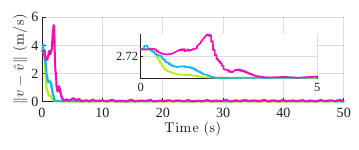}}
    \hfill
     \subfloat[\label{fig:eta-error}]{%
    \includegraphics[width=0.32\textwidth]{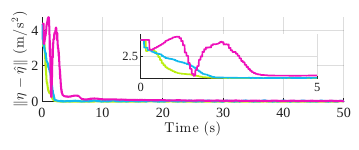}}\\[-0.1em]
     \subfloat[\label{fig:roll-error}]{%
    \includegraphics[width=0.32\textwidth]{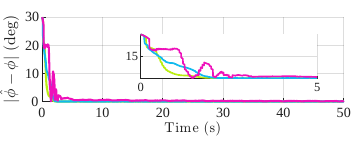}}
    \hfill
    \subfloat[\label{fig:yaw-error}]{%
    \includegraphics[width=0.32\textwidth]{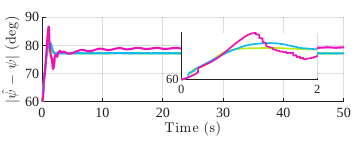}}
    \hfill
    \subfloat[\label{fig:pitch-error}]{%
    \includegraphics[width=0.32\textwidth]{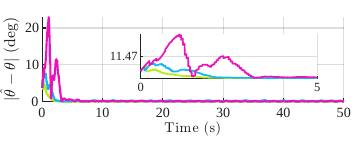}}
    \caption{Estimation errors of MVIO, MBVIO, and MBVIO-B:
    (a) aligned position, (b) body-frame velocity, (c) body-frame gravity, (d) roll (e) yaw, and (f) pitch errors.
    Green, blue, and magenta denote MVIO, MBVIO, and MBVIO-B, respectively.}
    \label{fig:estimation-errors}
\end{figure*}


\begin{figure}[htbp]
    \centering
    \includegraphics[width=0.48\columnwidth]
        {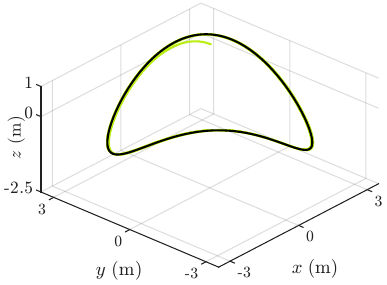}
    \hfill
    \includegraphics[width=0.48\columnwidth]
        {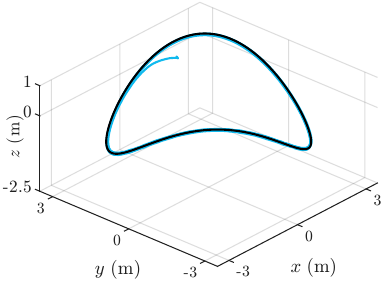}

    \includegraphics[width=0.48\columnwidth]
        {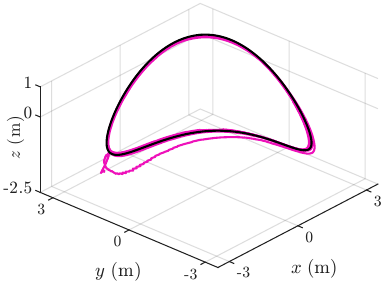}

    \caption{Ground-truth {\color{blue}(black)} and aligned trajectories for
    MVIO (green), MBVIO (blue), and
    MBVIO-B (magenta).}
    \label{fig:3D-traj}
\end{figure}
\begin{figure}[htbp]
    \centering
    \subfloat[\label{fig:accel-bias-error}]{%
    \includegraphics[width=0.475\columnwidth]{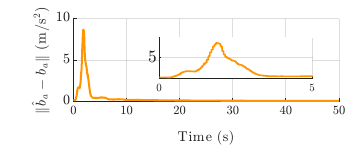}}
    \subfloat[\label{fig:accel-gyro-error}]{%
    \includegraphics[width=0.495\columnwidth]{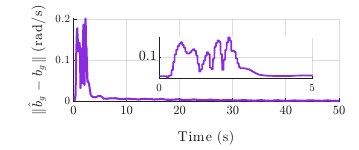}}
    \caption{Accelerometer and gyroscope bias estimation errors for MBVIO-B}
    \label{fig:Accel-gyro-bias-error}
\end{figure}

\section{Conclusion}
In this paper, we proposed three Riccati-based state observers, relying on IMU and monocular bearing measurements. The first observer estimates the body-frame linear velocity, gravity vector and landmarks with UGES guarantees. The second reduced observer estimates the body-frame linear velocity and gravity vector with UGES guarantees. The third observer estimates the IMU biases as well as the body-frame linear velocity and gravity vector, with local exponential stability guarantees. These estimates are subsequently used to recover the orientation almost globally, up to an unknown constant yaw offset, and the position up to an unknown constant translation. The key ingredients in the design of the reduced-order observers MBVIO and MBVIO-B are the introduction of the body-frame gravity vector as a new state to decouple the body-frame velocity from rotational dynamics, and the derivation of an appropriate output equation that relates the angular velocity, bearing measurements and their time derivatives to the body-frame linear velocity.\\
The price to pay for the derivation of the reduced-order observers (MBVIO and MBVIO-B) is the use of the bearings time-derivatives in the output equation. In simulations, reliably estimating the bearing time derivative from successive image measurements proved challenging, particularly in the presence of pixel noise. A two-frame backward-difference approximation did not yield satisfactory results. We therefore adopted Savitzky–Golay polynomial differentiation \cite{SGolay}, using endpoint differentiation over three consecutive frames for MBVIO and midpoint differentiation over five consecutive frames for MBVIO-B.
Future work will focus on evaluating the proposed observers on visual-inertial datasets and addressing the practical challenges associated with estimating and using bearing time-derivatives.


\begin{appendix}\label{Appendix}

\subsection{Proof of Proposition \ref{Proposition1}}\label{proof_uni_obs}

This proof is motivated by the proof of \cite[Lemma 1]{wangACC25}.
First, note that $A(t)$ and $C(t)$ given in~\eqref{eq:ABC_matrices_mixed_bearings} are continuous and uniformly bounded as long as $\omega(t)$ and $\dot{b}_i(t)$, $i=1,\hdots, N$, are continuous and uniformly bounded. Now we proceed to show that \eqref{UO} is satisfied as long as \eqref{PE_cond} is satisfied.
We rewrite the matrix $A(t)$ as  $  A(t) = \bar{A} - \Omega(t)$ with
\begin{equation}
    \bar{A} := \begin{bmatrix}
        0_{3\times3} & I_3 \\
        0_{3\times3} & 0_{3\times3}
    \end{bmatrix},
    \quad \Omega(t) := \blkdiag(\omega^\times (t), \omega^\times (t)).
\end{equation}
First, we show that the explicit form of the state-transition matrix associated with $A(t)$ is given by
\begin{equation}\label{eq:state-transition_prop2}
    {\Phi}(t,\tau) = \mathcal{R}(t)\bar{\Phi}(t,\tau) \mathcal{R}^\top (\tau),
\end{equation}
 where $\bar{\Phi}(t,\tau):= \exp{(\bar{A}(t-\tau))}$ is the state-transition matrix associated with $\bar{A}$, and  the block diagonal matrix $\mathcal{R}(t):= \blkdiag(R^\top, R^\top)$.

 Taking the time derivative of \eqref{eq:state-transition_prop2}, one obtains
\begin{equation}
     \frac{d}{dt} \Phi(t,\tau) = \dot{\mathcal{R}}(t) \bar{\Phi}(t,\tau) \mathcal{R}^\top (\tau) + \mathcal{R}(t) \bar{A} \bar{\Phi}(t,\tau) \mathcal{R}^\top(\tau).
 \end{equation}
 Using the facts $\dot{\mathcal{R}}(t) = - \Omega(t)\mathcal{R}(t)$ and $\mathcal{R}(t) \bar{A} = \bar{A}\mathcal{R}(t)$, one gets
 \begin{equation}
    \frac{d}{dt} \Phi(t,\tau) = A(t) \Phi(t,\tau).
 \end{equation}
Moreover, $\Phi(t,\tau)$ satisfies $\Phi(t,t)=I_{6}$, $\Phi(t_3,t_2)\,\Phi(t_2,t_1)=\Phi(t_3,t_1)$ for all $t_1,t_2,t_3\ge 0$, and  $\Phi(\tau,t)=\Phi(t,\tau)^{-1}$. Consequently, $\Phi(t,\tau)$ defined in \eqref{eq:state-transition_prop2} is the state-transition matrix associated with $A(t)$.

Substituting \eqref{eq:state-transition_prop2} in \eqref{obs_gram}, one has
\begin{equation}
    W_0 (t,t+\delta) =
       \mathcal{R}(t) \bar{W}_0  (t,t+\delta)\mathcal{R}^\top (t),
\end{equation}
where $\bar{W}_0(t,t+\delta):= \int_t ^ {t+\delta} \bar{\Phi}^\top (\tau,t) \bar{C}^\top(\tau) \bar{C}(\tau) \bar{\Phi}(\tau,t) d\tau$, and $\bar{C}(t):=R(t)C(t)\mathcal{R}(t)$. Notice that $\bar{W}_0$ is the observability Gramian associated with the pair $(\bar{A}, \bar{C}(t))$.

It is clear that if the pair $(\bar{A}, \bar{C}(t))$ is uniformly observable, \ie, there exist $\bar{\delta},\bar{\mu}>0$, such that for all $t\geq 0$, $\bar{W}_0(t,t+\bar{\delta})\geq \bar{\mu}I_6$, then for any $\delta\geq \bar{\delta}$, and $\mu\in(0,\bar{\mu}]  $, one has 
\begin{equation}
    {W}_0(t,t+{\delta})\geq \bar{\mu}\mathcal{R}(t)\mathcal{R}^\top(t)\geq \mu I_6.
\end{equation}

Consequently, the uniform observability of the pair $(\bar{A}, \bar{C}(t))$ implies the uniform observability of the pair $(A(t), C(t))$. Hence, it is sufficient to show that the pair  $(\bar{A}, \bar{C}(t))$ is uniformly observable.

Now, we prove the uniform observability of the pair $(\bar{A}, \bar{C}(t))$.
 We now establish the uniform observability of the pair $(\bar{A}, \bar{C}(t))$ using lemma \cite[Lemma 1]{NL_observer_V-INS}. Under the assumptions of Proposition \ref{Proposition1}, $\bar{C}(t)$ is continuous and bounded. Moreover, $\bar{A}$ is constant and satisfies $\bar{A}^2 = 0$, hence $\bar{A}$ is nilpotent of index two, and all its eigenvalues are zero. Therefore, the conditions required to apply lemma \cite[Lemma 1]{NL_observer_V-INS} are satisfied. It is then sufficient to show that there exist strictly positive constants $\bar{\mu}, \bar{\delta}$ such that, for all $t\geq0$
\begin{equation}\label{OO}
     \int _t ^ {t+\bar{\delta}} \mathcal{O}^\top(\tau)\mathcal{O}(\tau) d\tau \geq \bar{\mu}I_6,
\end{equation}
where 
\[
\mathcal{O}(t):= [ \bar{C}^\top(t), (\bar{C}(t) \bar{A})^\top ] ^\top.
\]

Using \eqref{eq:ABC_matrices_mixed_bearings}, $\bar{C}(t)$ can be expressed as $\bar{C}(t) = [\tilde{M}(t), 0_{3\times3}]$, where
$ \tilde{M}(t)= R(t) R_c M(t) R_c ^\top R^\top (t)= \sum _{i=1}^ N \gamma_i \bar{q}_i\bar{q}_i ^\top.$ 
It follows that
\begin{equation}
\begin{aligned}
     \mathcal{O}^\top(t) \mathcal{O}(t) &= \begin{bmatrix}
        \tilde{M}^\top (t) \tilde{M}(t) & 0_{3\times3}\\
        0_{3\times3} & \tilde{M}^\top (t) \tilde{M}(t)
    \end{bmatrix} 
\end{aligned}
\end{equation}

Integrating on both sides with respect to time over $[t,t+\bar{\delta}]$, one obtains
\begin{equation} 
\begin{aligned}
     &\int _t ^ {t+\bar{\delta}} \mathcal{O}^\top (\tau) \mathcal{O} (\tau) d\tau \\ &= \begin{bmatrix}
        \int _t ^ {t+\bar{\delta} } \tilde{M}(\tau)^\top \tilde{M}(\tau) d\tau & 0_{3\times3} \\
        0_{3\times3} & \int _t ^ {t+\bar{\delta} } \tilde{M}(\tau)^\top \tilde{M}(\tau)d\tau
    \end{bmatrix}
\end{aligned}
\end{equation}
Assuming that condition \eqref{PE_cond} holds, and using Cauchy-Schwarz inequality, one has 
\begin{equation}
    \int _t ^{t+\delta_M} \tilde{M}^\top(\tau) \tilde{M}(\tau) d \tau \geq \frac{\mu_M^2}{\delta_M} I_3.
\end{equation}

Therefore, \eqref{OO} holds with $\bar{\delta}=\delta_M$ and $\bar{\mu}=\frac{\mu_M ^2}{\delta_M}$.
Consequently, the pair $(\bar{A}, \bar{C}(t))$ is uniformly observable. Hence, the uniform observability of the pair $(A(t), C(t))$ follows.

\subsection{Proof of Proposition \ref{Prop_suff_cond_excit} } \label{proof_prop_suff_cond}
 In view of \eqref{q_i} and the fact that $\bar{q}_i(t):= R(t)R_cq_i(t)$, one has $\bar{q}_i=-\frac{1}{l_i}(\bar{R}b_i \times \bar{R} v_c)$, with $\bar{R}(t)=R(t)R_c$. Then, condition \eqref{PE_cond} can be rewritten as
\begin{equation}\label{PE_cond_modified}
        \int _t ^ {t+\delta_M} (\bar{R}v_c^\times) \left( \sum _{i=1} ^ N  \frac{\gamma_i}{l_i^2} b_i b_i^\top \right)(\bar{R}v_c^\times)^\top d\tau \geq \mu_MI_3.
     \end{equation}
Assume that there exists $\beta>0$ such that 
\begin{equation}\label{bearing_excitation}
\pi(u_v)\left(\sum _{i=1} ^ N  \frac{\gamma_i}{l_i^2} b_i(t) b_i^\top(t)\right) \pi(u_v) \geq \beta \pi(u_v), ~\forall t\geq0,
\end{equation}
where $u_v=\frac{v_c}{\|v_c\|}$ whenever $v_c\not=0$. Since $v_c^\times u_v=0$, one has $v_c^\times\pi(u_v)=v_c^\times$ and hence, 
\begin{align*}\label{PE_cond_modified2}
& \int _t ^ {t+\delta_M} (\bar{R}v_c^\times) \left( \sum _{i=1} ^ N  \frac{\gamma_i}{l_i^2} b_i b_i^\top \right)(\bar{R}v_c^\times)^\top d\tau \\
&=\int _t ^ {t+\delta_M} (\bar{R}v_c^\times) \pi(u_v) \left( \sum _{i=1} ^ N  \frac{\gamma_i}{l_i^2} b_i b_i^\top \right)\pi(u_v)(\bar{R}v_c^\times)^\top d\tau \\
&\geq -\beta \int _t ^ {t+\delta_M} \bar{R}(v_c^\times)^2 \bar{R}^\top d\tau=\beta \int _t ^ {t+\delta_M} (\|\bar{v}_c\|^2 I- \bar{v}_c\bar{v}_c^\top) d\tau \\
 \end{align*}
with $\bar{v}_c=\bar{R}v_c$. 
Assume that there exists $\zeta>0$ such that 
\begin{equation}\label{excitation_v}
\int _t ^ {t+\delta_M} (\|\bar{v}_c\|^2 I- \bar{v}_c\bar{v}_c^\top) d\tau \geq \zeta I, ~\forall t\geq 0,
\end{equation}
is satisfied as well as \eqref{bearing_excitation}, then condition \eqref{PE_cond} is satisfied with $\mu_M=\beta \zeta $.\\
Now, let us look at condition \eqref{excitation_v}. For any unit-vector $z$, one has 
\begin{equation}\label{excitation_v2}
 \begin{aligned}
&z^\top \left(\int _t ^ {t+\delta_M} (\|\bar{v}_c\|^2 I- \bar{v}_c\bar{v}_c^\top) d\tau\right)z\\
&=\int _t ^ {t+\delta_M} ( \|\bar{v}_c\|^2- (z^\top \bar{v}_c)^2) d \tau=\int _t ^ {t+\delta_M}\|z \times \bar{v}_c(\tau)\|^2 d \tau.
 \end{aligned}
\end{equation}
Consequently, condition \eqref{excitation_v} is satisfied as long as $\bar{v}_c$ does not remain aligned with a fixed direction throughout any interval of length $\delta_M$.\\
Now let us look at condition \eqref{bearing_excitation}. Let $\bar{z}$ be a unit vector orthogonal to $u_v$, that is $\pi(u_v)\bar{z}=\bar{z}$ and consequently, for every $t$ such that $v_c(t)\not=0$, condition \eqref{bearing_excitation} becomes
\begin{equation}
\sum _{i=1} ^ N  \frac{\gamma_i}{l_i^2} (b_i^\top \bar{z}(t))^2 \geq \beta , ~\forall t\geq0,
\end{equation}
which is satisfied as long as no direction in the plane orthogonal to $u_v$ (consequently orthogonal to $v_c$) can be orthogonal to all bearings $b_i$, $ i \in \mathbb{I}$.





\end{appendix}

\bibliographystyle{IEEEtran}
\bibliography{reference}
\end{document}